\documentclass{article}

\renewcommand{\parallel}{{|\!|}}

\usepackage{calc,graphicx,
url,
amsmath,
amssymb,
amsthm,
latexsym,xcolor,xspace,multicol}

\def\N{\mbox{I\hspace{-.15em}N}}

\newcommand\DCLPC[0]{{\textsf{DCL{-}PC}}\xspace}
\newcommand\DLPA[0]{{\textsf{DL{-}PA}}\xspace}
\newcommand\PDL[0]{{\textsf{PDL}}\xspace}

\newcommand{\Agt}     { \mathbb{A} }

\newcommand{\lbox}[1]{[ #1 ]}
\newcommand{\ldia}[1]{{\langle #1 \rangle}}

\newcommand{\limp}{\rightarrow}
\newcommand{\tuple}[1]{( #1 )}
\newcommand{\assgntop}[1]{ {+}{#1} }
\newcommand{\assgnbot}[1]{ {-}{#1} }
\newcommand{\eqdef}{\stackrel{\mathtt{def}}{=} }
\newcommand{\lldot}{{\ll\hspace{-.45em}\cdot}}

\newcommand{\transfer}[3]   { #1 \, { \leadsto_{#2} } #3}

\newcommand{\eloiseTurn}{\mathsf{elo}}
\newcommand{\eloiseLoses}{\mathsf{nowin}}
\newcommand{\oneMove}{\mathsf{move}}
\newcommand{\modelM}{\mathcal{M}}

\newtheorem{theorem}{Theorem}
\newtheorem{proposition}{Proposition}
\newtheorem{corollary}{Corollary}
\newtheorem{lemma}{Lemma}
\newtheorem{claim}{Claim}

\newcommand{\algofunction}{\textbf{function }}

\newcommand{\algoprocedure}{\textbf{procedure }}

\newcommand{\algoendfunction}{\textbf{endFunction }}

\newcommand{\algofor}{\textbf{for }}
\newcommand{\algodo}{\textbf{do }}
\newcommand{\algoendfor}{\textbf{endFor }}

\definecolor{algocommentbackgroundcolor}{rgb}{1,1,0.5}

\newcommand{\algoif}{\textbf{if }}
\newcommand{\algothen}{\textbf{then }}
\newcommand{\algoelse}{\textbf{else }}
\newcommand{\algoendif}{\textbf{endIf }}

\newcommand{\algomatch}{\textbf{match }}
\newcommand{\algoendmatch}{\textbf{endMatch }}

\newcommand{\algocase}{\textbf{case }}

\newcommand{\algoreturn}{\textbf{return }}

\newlength{\algoindentlongueur}
\newcommand{\algoindent}{\hspace*{\algoindentlongueur}}

\newlength{\algoindentavantvrulelongueur}
\newcommand{\algoindentavantvrule}{\hspace*{\algoindentavantvrulelongueur}}

\newlength{\dummy}

\newsavebox{\frameminipageboiteavecunnomsuperlongdesortequonnepuissepaslereutiliser}
\newenvironment{frameminipage}[2][c]{%
\begin{lrbox}{\frameminipageboiteavecunnomsuperlongdesortequonnepuissepaslereutiliser}%
\begin{minipage}[#1]{#2}%
} {%
\end{minipage}%
\end{lrbox}%
\framebox{\usebox{\frameminipageboiteavecunnomsuperlongdesortequonnepuissepaslereutiliser}}%
}

\newenvironment{algo} {
  \begin{frameminipage}{\textwidth-1cm}
} {
  \end{frameminipage}
}

\newenvironment{algobloc}{\setlength{\dummy}{\linewidth}\addtolength{\dummy}{- \algoindentlongueur}\addtolength{\dummy}{- \algoindentavantvrulelongueur}\algoindentavantvrule\vrule\algoindent\begin{minipage}{\dummy}}{\end{minipage}}

\newenvironment{algoblocfunction}[1]
{\algofunction #1 \\  \begin{algobloc}}
{\end{algobloc} \algoendfunction}

\date{}

\title{\DLPA and \DCLPC: model checking and satisfiability problem are indeed in PSPACE}
\author{Philippe Balbiani\thanks{Universit\'e de Toulouse, CNRS, IRIT, F-31062 Toulouse, France}\\ Andreas Herzig\thanks{Universit\'e de Toulouse, CNRS, IRIT, F-31062 Toulouse, France}\\
Fran\c{c}ois Schwarzentruber\thanks{ENS Rennes, Campus de Ker Lann, 35170 BRUZ, France}\\
Nicolas Troquard\thanks{ISTC--CNR, Trento, Italy. LACL, Universit\'e Paris-Est Cr\'eteil, France.}}

\begin{document}

\maketitle
  
  \begin{abstract}
We prove that the model checking and the satisfiability problem of both Dynamic Logic of Propositional Assignments \DLPA and Coalition Logic of Propositional Control and Delegation \DCLPC are in PSPACE. We explain why the proof of EXPTIME-hardness of the model checking problem of \DLPA presented in \cite[Thm~$4$]{DBLP:conf/lics/BalbianiHT13} is false.
We also explain why the proof of membership in PSPACE of the model checking problem  of \DCLPC given in \cite[Thm.~$4$]{DBLP:journals/jair/HoekWW10} is wrong.

 Coalition Logic of Propositional Control and Delegation. Model checking. Satisfiability. PSPACE.

 \end{abstract}

\section{Introduction}

Balbiani~{\it et al}~\cite{DBLP:conf/lics/BalbianiHT13} study a
variant of \PDL called Dynamic Logic of Propositional Assignments
(\DLPA).
The latter was introduced in~\cite{DBLP:conf/ijcai/HerzigLMT11} and is a fragment of 
Tiomkin and Makowsky's extension of \PDL by assignments \cite{TiomkinMakowsky85}. 
It is said to be well-behaved
because unlike \PDL, it is compact, has the interpolation property,
and the Kleene star can be eliminated.  The logic was partly inspired
by the logic of delegation and propositional control \DCLPC presented
in
\cite{DBLP:journals/jair/HoekWW10}. In~\cite{DBLP:conf/lics/BalbianiHT13},
polynomial translations from \DCLPC to \DLPA and back are proposed.

\medskip

Between the papers \cite{DBLP:journals/jair/HoekWW10},
\cite{DBLP:conf/ijcai/HerzigLMT11} and
\cite{DBLP:conf/lics/BalbianiHT13}, there have been conflicting
results about the complexity of decision problems for \DLPA and
\DCLPC, satisfiability checking and model checking. There have also been
inadequate proofs for true theorem statements, and there have been
wrong proofs for wrong theorem statements.
The aim of the present paper is to set the record straight. Specifically:\footnote{
We recall that PSPACE = NPSPACE (\cite{Savitch1970177}, \cite{sipser2006introduction}) and APSPACE = EXPTIME (\cite{Chandra:1981:ALT:322234.322243}). In this paper we assume that PSPACE is different from EXPTIME. If they are equal the whole discussion ends up being a non-issue.
}
\begin{itemize}
\item The proof in~\cite{DBLP:journals/jair/HoekWW10} that \DCLPC model checking is in PSPACE is inadequate. It only proves that it is in EXPTIME. A consequence is also that the proof that \DCLPC satisfiability checking is in PSPACE is inadequate, too.
\item Following the same proof strategy, the proof in \cite{DBLP:conf/ijcai/HerzigLMT11} that \DLPA model checking is in PSPACE is inadequate. It only proves that it is in EXPTIME.\footnote{
The error is in the published version and is signaled on the website of the conference
\url{http://ijcai.org/papers11}.
}
\item The proof in~\cite{DBLP:conf/lics/BalbianiHT13} that \DLPA model checking is EXPTIME-hard is wrong: more precisely, the statement of ~\cite[Thm~$4$]{DBLP:conf/lics/BalbianiHT13} is wrong.
\item The model checking problem and the satisfiability checking problem of \DLPA are both PSPACE-complete.
\item The model checking problem and the satisfiability checking problem of \DCLPC are both PSPACE-complete.
\end{itemize}

\section{Two dynamic logics}

We present \DLPA and \DCLPC which are two interconnected dynamic logics.

\subsection{Dynamic logic of propositional assignments \DLPA}

\paragraph{Syntax}
Let $PV$ be a countable set of propositional variables (with typical members noted $p$, $q$, etc).
The set $\mathsf{Pgm}(PV)$ of all programs (with typical members noted $\alpha$, $\beta$, etc) and the set $\mathsf{Fml}(PV)$ of all formulas (with typical members noted $\phi$, $\psi$, etc) are inductively defined as follows:
\[
\begin{array}{lccccccccccccccc}
\alpha & ::= & +p & \mid & -p & \mid & (\alpha;\alpha) & \mid & (\alpha\cup\alpha) & \mid & \alpha^{\star} & \mid & \phi?\\
\phi & ::=  & p & \mid & \bot & \mid & \lbrack\alpha\rbrack\phi
\end{array}
\]
We define the other Boolean constructs as usual: 
$\neg\phi = \lbrack\phi?\rbrack\bot$, 
$(\phi\rightarrow\psi) = \lbrack\phi?\rbrack\psi$, etc.
The formula $\langle\alpha\rangle\phi$ is obtained as an abbreviation: 
$\langle\alpha\rangle\phi = \neg\lbrack\alpha\rbrack\neg\phi$. We write $\alpha^d$ for the sequence of $\alpha$ repeated $d$ times.
We adopt the standard rules for omission of the parentheses.
Let us consider an enumeration $p_{1},p_{2},\ldots$ of $PV$. 
%The length of a program $\alpha$, in symbols $len(\phi)$, and the length of a formula $\phi$, in symbols $len(\phi)$, are the positive integers inductively defined as follows:
%\begin{align*}
%len(+p_{i}) &=\lfloor\log i\rfloor+3\\
%len(-p_{i}) &=\lfloor\log i\rfloor+3\\
%len(\alpha;\beta) &=len(\alpha)+len(\beta)+3\\
%len(\alpha\cup\beta) &=len(\alpha)+len(\beta)+3\\
%len(\alpha^{\star}) &=len(\alpha)+1\\
%len(\phi?) &=len(\phi)+1\\
%len(p_{i}) &=\lfloor\log i\rfloor+2\\
%len(\bot) &=1\\
%len(\lbrack\alpha\rbrack\phi) &=len(\alpha)+len(\phi)+2
%\end{align*}
Program ``$+p$'' makes proposition $p$ true and program ``$-p$'' makes proposition $p$ false.
The number of symbol occurrences in program $\alpha$ and formula $\phi$ are respectively
 noted $len(\alpha)$ and $len(\phi)$.

\paragraph{Semantics}
A valuation is a subset of $PV$, with typical elements $U$, $V$, etc.
We inductively define the value of a program $\alpha$, in symbols $\parallel\alpha\parallel$, and the value of a formula $\phi$, in symbols $\parallel\phi\parallel$, as follows:
\begin{align*}
\parallel+p\parallel \ =\ & \{(U,V)\ : \ V=U\cup\{p\}\}\\
\parallel-p\parallel \ =\ & \{(U,V)\ : \ V=U\setminus\{p\}\}\\
\parallel\alpha;\beta\parallel \ =\ & \{(U,V)\ : \ \text{there exists } W\subseteq PV\text{ such that }\\
&\   (U,W)\in\parallel\alpha\parallel \text{ and } (W,V)\in\parallel\beta\parallel\}\\
\parallel\alpha\cup\beta\parallel \ =\ & \parallel\alpha\parallel\cup\parallel\beta\parallel\\
\parallel\alpha^{\star}\parallel \ = \ & \{(U,V)\ : \ \text{there exist } n\in\N \text{ and } W_{0},\ldots,W_{n}\subseteq PV \text{ such that } \\
 &\ U=W_{0}, (W_{0},W_{1})\in\parallel\alpha\parallel, \ldots, (W_{n-1},W_{n})\in\parallel\alpha\parallel %\\ &\ 
\text{ and } W_{n}=V\}\\
\parallel\phi?\parallel \ =\ & \{(U,V) \ : \ U =V \text{ and } V\in\parallel\phi\parallel\}\\
\parallel p\parallel \ =\ & \{U\ : \ p\in U\}\\
\parallel\bot\parallel \ =\ & \emptyset\\
\parallel\lbrack\alpha\rbrack\phi\parallel \ = \ & \{U \ : \text{ for all } V\subseteq PV, \text{ if } (U,V)\in\parallel\alpha\parallel, \text{ then } V\in\parallel\phi\parallel\}
\end{align*}

%\begin{lemma}
It follows that
%for all programs $\alpha$ and for all formulas $\phi$, 
$\parallel\langle\alpha\rangle\phi\parallel \ =
\{ U$: there exists $V\subseteq PV$ such that $(U,V)\in\parallel\alpha\parallel$ and $V\in\parallel\phi\parallel \}$.
%\end{lemma}

\subsection{Coalition logic of propositional control and delegation \DCLPC}
\label{sec:clpc}
\label{sec:dclpc}

Coalition Logic of Propositional Control and Delegation (\DCLPC) is a
logic of \emph{agency}. Let $PV$ be a countable set of propositional
variables and $\Agt$ be a finite set of \emph{agents}.

The models of \DCLPC---models of propositional control---are couples
$(V,\xi)$ where $V$ is a subset of $PV$ and $\xi$ maps each
propositional variable to one agent in $\Agt$. The function $\xi$ is a
control function. Intuitively, for each proposition $p$, the object
$\xi(p)$ denotes the one and only one agent controlling it. Saying
that the agent $\xi(p)$ controls $p$, we mean that $\xi(p)$ can set
$p$ to true and can set $p$ to false.

\medskip

The language of \DCLPC extends propositional logic with two families
of modalities. One type of modalities is reminiscent of dynamic
logics, and thus we have a two-sorted language. In the following
grammar, $i, j \in \Agt$, and $p \in PV$.
\[
\begin{array}{lccccccccccccccc}
\pi & ::= & \transfer{i}{p}{j} & \mid & (\pi;\pi) & \mid & (\pi\cup\pi) & \mid & \pi^{\star} & \mid & \phi?\\
\phi & ::=  & p & \mid & \bot & \mid & \Diamond_i \phi& \mid & \ldia{\pi}\phi 
\end{array}
\]
We adopt the standard abbreviations.

\medskip

To differentiate the truth values of \DCLPC from those of \DLPA, we
will denote the value of \DCLPC programs and \DCLPC formulas by
$\parallel . \parallel^\#$.

Atomic delegation programs are of the form $\transfer{i}{p}{j}$ and
are read ``$i$ transfers his control over $p$ to $j$''.  The intuition
is that $\transfer{i}{p}{j}$ is applicable when $i$ controls $p$ and
that it changes the control function $\xi$ such that $j$ gets control
over $p$ (and $i$ looses it, control being exclusive).  Complex
delegation programs are defined by means of the standard \PDL
operators.
The interpretation of a delegation program is a binary relation on the
set of models of propositional control over $PV$ and $\Agt$.  For
atomic programs we have:
\begin{align*}
\parallel{\transfer{i}{p}{j}}\parallel^\# \ =\  \big \{ \, \tuple{(V,\xi),(V,\xi')} \ : \ &
\xi(p) = i, \
\xi'(p) = j, \text{ and } \\&
\xi(q) = \xi'(q) \text{ for } q \neq p \, \big \}
\end{align*}
The interpretation of complex programs is as usual.

The interpretation of \DCLPC formulas is a subset of models of
propositional control over $PV$ and $\Agt$. 
\begin{align*}
\parallel p\parallel^\# \ =\ & \{(V,\xi)\ : \ p\in V\}
\end{align*}
The interpretation of $\ldia{\pi}\varphi$ is:
\begin{align*}
\parallel{\ldia{\pi} \phi}\parallel^\# \ =\ 
\big \{ \, (V,\xi) \ : \ &\text{there is } (U,\xi') \text{ such that } \\ & 
\tuple{(V,\xi),(U,\xi')}  \in \parallel{\pi}\parallel^\#
\text{ and } (U,\xi') \in \parallel{\phi}\parallel^\# \, \big \}
\end{align*}
The modality $\Diamond_i$ allows one to talk about what an
agent $i$ is able to do by changing the truth value of the
propositional variables under its control. 
\begin{align*}
\parallel{\Diamond_i \phi}\parallel^\# \ = \{\, (V,\xi) \ : \ &\text{there is } U \text{ such that } 
(U,\xi) \in \parallel{\phi}\parallel^\# \text{ and } \\& 
\text{for every } p,  
\text{if } \xi(p) \not = i \text{ then } p \in V \text{ iff } p \in U \, \}
\end{align*}
The interpretation of complex formulas is as usual.

\subsection{Connection}

As announced the two dynamic logics reviewed here are
interconnected. In particular, we can apply the algorithms for the
decision problems of \DLPA to solve the the decision problems of
\DCLPC. Of concern here are four decision problems:

\begin{itemize}

\item \DLPA-model checking ($\mathit{MC}$):
\begin{description}
\item[input:] a valuation $U$, and a formula $\phi \in \mathsf{Fml}(PV)$,
\item[output:] \emph{yes} if $U\in \parallel\phi\parallel$, \emph{no} otherwise.
\end{description}

\item \DLPA-satisfiability ($\mathit{SAT}$):
\begin{description}
\item[input:] a formula $\phi \in \mathsf{Fml}(PV)$,
\item[output:] \emph{yes} if $\parallel\phi\parallel \not = \emptyset$, \emph{no} otherwise.
\end{description}

\item \DCLPC-model checking
\begin{description}
\item[input:] a model of propositional control $(U,\xi)$, and a \DCLPC formula,
\item[output:] \emph{yes} if $(U,\xi)\in \parallel\phi\parallel^\#$, \emph{no} otherwise.
\end{description}

\item \DCLPC-satisfiability:
\begin{description}
\item[input:] a \DCLPC formula $\phi$,
\item[output:] \emph{yes} if $\parallel\phi\parallel^\# \not = \emptyset$, \emph{no} otherwise.
\end{description}
\end{itemize}

\begin{theorem}[{\cite[Section VIII]{DBLP:conf/lics/BalbianiHT13}}]
\label{prop:delClPc-embed}
There is a polynomial reduction of \DCLPC-model checking into
\DLPA-model checking. There is a polynomial reduction of
\DCLPC-satisfiability into \DLPA-satisfiability.
\end{theorem}
Hence, the complexity upper bound for a problem of \DLPA will transfer
polynomially to a complexity upper bound for the corresponding problem
of \DCLPC.

\section{Issue in the proof of~\cite[Thm~$4$]{DBLP:conf/lics/BalbianiHT13}}

Theorem~$4$ in~\cite{DBLP:conf/lics/BalbianiHT13} wrongly states that
$\mathit{MC}$ and $\mathit{SAT}$ are EXPTIME-hard. The source of the
problem lies in \cite[Lemma~$1$]{DBLP:conf/lics/BalbianiHT13} which
wrongly states that $\mathit{MC}$ is EXPTIME-hard, proposing an
inadequate argument for establishing the existence of a
logarithmic-space reduction of the problem
PEEK-${G_5}$~\cite{stockemeyer-chandra-79} into $\mathit{MC}$. The
claim about $\mathit{SAT}$ then comes from an actual logarithmic-space
reduction of the problem $\mathit{MC}$ into $\mathit{SAT}$.

\medskip

This section concentrates on the issue with the reduction of the
problem PEEK-${G_5}$ into $\mathit{MC}$. 

\emph{An instance of Peek} is a tuple $PE = \tuple{X_E, X_A, \Phi, V_0, \tau}$
where
%	\begin{itemize}
%	\item 
	$X_E$ and $X_A$ are finite sets of propositional variables 
	such that $X_E \cap X_A = \emptyset$,
	the idea being that Player~$E$ controls the variables in $X_E$ and 
	Player~$A$ controls the variables in $X_A$;
%	\item 
	$\Phi$ is a propositional formula over $X_E \cup X_A$;
%	\item 
	$V_0 \subseteq X_E \cup X_A$ indicates which variables are
	  currently true;
%	\item 
	$\tau$ is either~$A$ or $E$, indicating which player makes the next move.
%	\end{itemize}

Informally, each instance $PE = \tuple{X_E, X_A, \Phi, V_0, \tau}$ of
Peek is played as follows.  Agents' turns strictly alternate.  At
their respective turn, Player~$E$ (resp.\ $A$) \emph{moves} by
changing the truth value of at most one variable of $X_E$
(resp.\ $X_A$) in the current valuation, either adding or withdrawing
it from the valuation.  The game ends when $\Phi$ first becomes true,
in which case we say that Player~$E$ wins.  We say that
\emph{Player~$E$ has a winning strategy in $PE$} if she can make a
sequence of moves at her turns that ensures to eventually win whatever
the moves made by Player~$A$ at his turn.

The decision problem PEEK-$G_5$ takes as input an instance\linebreak $PE =
\tuple{X_E, X_A, \Phi, V_0, \tau }$ of Peek; It outputs \emph{yes},
when Player~$E$ has a winning strategy in $PE$ and \emph{no}
otherwise. PEEK-$G_5$ is EXPTIME-complete~\cite{stockemeyer-chandra-79}.

\medskip

In~\cite[Lemma~$1$]{DBLP:conf/lics/BalbianiHT13}, it was stated that
the problem PEEK-$G_5$ on the instance $PE = \tuple{X_E,X_A, \Phi, V_0, \tau}$ returns \emph{no} if and only if $\mathit{MC}$ return \emph{yes} on the instance $(V_{PE},\varphi_{PE})$, where:
\small
\begin{align*}
V_{PE} &\ \eqdef \begin{cases}V_{0}\cup\{\eloiseLoses\} & \text{ , when }  \tau=A\\
V_{0}\cup\{\eloiseLoses, \eloiseTurn\} & \text{ , when } \tau=E\end{cases}\\
\mathsf{moveE} &\ \eqdef\  	\eloiseTurn? ; 
					\bigcup_{x \in X_E} (\assgnbot{x} \cup \assgntop{x}) ; 
					\assgnbot{\eloiseTurn}
\\
\mathsf{moveA} &\ \eqdef\  	\lnot \eloiseTurn? ; 
					\cup_{y \in X_A} (\assgnbot{y} \cup \assgntop{y}) ; 
					\assgntop{\eloiseTurn}
\\
\oneMove &\ \eqdef\  	\left ( {\sf moveE} \cup {\sf moveA} \right ) ; 
				\left ( (\Phi ? ; \assgnbot{\eloiseLoses}) \cup \lnot \Phi ? \right ) \\
\varphi_{PE} &\ \eqdef \lbox{\oneMove^\ast} \big ( \eloiseLoses \limp ( \lnot \Phi \land 
( \eloiseTurn 	  \limp \lbox{\oneMove} \eloiseLoses ) \land%\\
%& \hspace{45mm}  
( \lnot\eloiseTurn \limp \ldia{\oneMove} \eloiseLoses ))      \big )
\end{align*}
\normalsize
This is incorrect. 
For the anecdote, the mistake was found when one of us figured that if the reduction were actually working, a similar reduction could be done from PEEK-$G_5$ into the problem of model checking CTL formulas over NuSMV models, which is known to be in PSPACE. The implementation of it and the checking of a simple instance indicated the mistake.\footnote{The NuSMV file can be found at this URL \url{http://www.loa.istc.cnr.it/personal/troquard/SOFTWARES/error-peekdlpa.smv} and its listing is presented in the appendix.}
The instance of Peek considered was $PE=(X_{E},X_{A},\Phi,V_{0},\tau)$,
where $X_{E}=\{p\}$, $X_{A}=\{q,r\}$, $\Phi=p\wedge q$,
$V_{0}=\emptyset$ and $\tau=A$.  Clearly, if $A$ never adds $q$ to the
valuation $V_0$, then $\Phi$ cannot ever be true. Since $\tau =
A$, this means that $E$ has no winning strategy in the game, and PEEK-$G_5$ returns \emph{no} on this instance.
However, the problem $\mathit{MC}$ also returns \emph{no} on the instance $(V_{PE},\varphi_{PE})$, establishing a counter-example to \cite[Lemma~$1$]{DBLP:conf/lics/BalbianiHT13}. 

Without this lemma, Proposition~$14$ in \cite{DBLP:conf/lics/BalbianiHT13} stating that $\mathit{MC}$ is EXPTIME-hard has no basis. In turn, Proposition~$15$ about $\mathit{SAT}$ being EXPTIME-hard has no basis either. Theorem~$4$ in \cite{DBLP:conf/lics/BalbianiHT13} is wrong if PSPACE $\neq$ EXPTIME.

\section{On the proof of~\cite[Thm.~$4$]{DBLP:journals/jair/HoekWW10} for PSPACE membership of \DCLPC model checking}
In~\cite{DBLP:journals/jair/HoekWW10}, the authors state that the
model checking problem for \DCLPC (w.r.t.\ direct models) is
PSPACE-complete. As we shall see later, the result is true in virtue of
the algorithm for solving model checking problem for \DLPA
(Section~\ref{sec:algorithm-MC}) and
Theorem~\ref{prop:delClPc-embed}. Nevertheless, the algorithm proposed
is alternating, not non-deterministic as claimed in the article. It
therefore only allows one to conclude that the \DCLPC model checking
problem is in APSPACE and not in PSPACE.  This was already pointed out
in~\cite{DBLP:conf/lics/BalbianiHT13}; we provide a more complete
explanation now.

Let us explain why the algorithm is alternating and not
non-deterministic. In fact their algorithm is of the following form. Algorithm `DCL-PCeval' of Figure~$8$,
line~$5$ in~\cite{DBLP:journals/jair/HoekWW10} negates the Boolean
result in the following way:

\begin{algo}
\begin{algoblocfunction}{DCL-PCeval($\phi, \modelM$)}

\algoif $\dots$ \algothen

\begin{algobloc}

 $ \vdots$
 
\end{algobloc} 
 \algoelse \algoif $\phi = \lnot \psi$ \algothen
 
 \begin{algobloc}

 \algoreturn \colorbox{red!20}{\textbf{not} DCL-PCeval($\psi, \modelM$)}

\end{algobloc} 

\algoelse

 \begin{algobloc}
 
 $\vdots$ (with a call to program-eval)

\end{algobloc}

\algoendif

\end{algoblocfunction}
\end{algo}

\medskip\noindent where `program-eval' (see Fig.~$7$
in~\cite{DBLP:journals/jair/HoekWW10}) explicitly mentions
non-deterministic choices. But negation implicitly dualizes the algorithm: it transforms
true, false, non-deterministic choice, and universal choices into
false, true, universal choices, and non-deterministic choice
respectively. So the algorithm is in fact alternating.\footnote{
Using the `return' instruction to return the Boolean result of a
function is perfectly correct in a \emph{deterministic}
algorithm. Nevertheless, when one writes non-deterministic algorithms
one should \emph{explicitly} use the `reject' and `accept'
instructions that respectively correspond to the rejection and the
acceptation state in a Turing machine. Negations are strictly
forbidden in a non-deterministic algorithm.
}

\section{A deterministic procedure for \DLPA-model checking and satisfiability problem}
\label{sec:algorithm-MC}

Our goal in this paper is to prove the following result.
\begin{proposition}\label{pro_pspace}
The \DLPA-model checking and satisfiability problem are in PSPACE.
\end{proposition}

Proposition~\ref{pro_pspace} will be obtained as a direct consequence of Proposition~\ref{pro_principal} and Claims~\ref{claim:claim1} and~\ref{claim:claim2}. 
As to $\mathit{SAT}$, 
one can check satisfiability of a formula $\phi$
by an algorithm which first guesses a valuation $v$ and then model-checks whether $v \models \phi$.
This algorithm works in nondeterministic polynomial space NPSPACE, and 
NPSPACE = PSPACE due to Savitch's Theorem.

Furthermore, by Theorem \ref{prop:delClPc-embed} we have:
\begin{corollary}
The \DCLPC-model checking and satisfiability problem are in PSPACE.
\end{corollary}

\subsection{Divide and conquer}

Divide and conquer is a familiar algorithmic design technique: for solving a problem, we cut it in several pieces, solve subproblems and combine their results. In the model checking problem for \DLPA, the subproblem to which we will apply divide and conquer is the following one:

\begin{description}
\item[input:] two valuations $U$, $V$, a program $\alpha$;
\item[output:] \emph{yes} if $(U, V) \in \parallel\alpha\parallel$, no otherwise.
\end{description}

This problem becomes tricky when $\alpha$ is of the form $\beta^*$. As we are concerned by a finite set of propositional variables, let say $k$ propositional variables, the cardinal of the set of valuations is $2^k$. Therefore, $(U, V) \in \parallel \beta^{*}\parallel$ is equivalent to $(U, V) \in \parallel \beta^N\parallel$ for $N \in \{0, \dots, 2^{k}-1\}$. In particular, if $N$ is even, $(U, V) \in \parallel \beta^N \parallel$ iff there exists $W$ such that $(U, W) \in \parallel \beta^{\frac N 2}\parallel$ and $(W, V) \in \parallel \beta^{\frac N 2}\parallel$. Thanks to divide and conquer, we are able to design an algorithm that works in polynomial space for the model checking problem in \DLPA.

Actually, the divide and conquer paradigm already appears in the proof of Savitch's theorem (\cite{Savitch1970177}, \cite{sipser2006introduction}). It has also been recently applied to prove the membership in PSPACE of the model checking of an epistemic formula dealing with agent cameras \cite{AAMAS2014}.

\subsection{Description of the algorithm}

Let us assume that the language only contains $k$ propositional variables.
In the sequel, sequences of bits are sequence of length $k$ whereas ``$+1$'' means ``$+1$ modulo $2^{k}$''.
Such sequences will be used to represent valuations.
More precisely, the valuation represented by a sequence $val$ of $k$ bits makes propositional variable $p_{i}$ true iff the $i$-th bit of $val$ is $1$.
Sequences of $k$ bits will also be used to represent integers in $\{0,\ldots,2^{k}-1\}$.
In this case, they will be noted by $d$, $e$, etc.
In the sequel, for all sequences $d,e$ of $k$ bits, $d<e$ will mean that the integer represented by the sequence $d$ is strictly smaller than the integer represented by the sequence $e$.
We define the deterministic Boolean function $REL$ taking as input a bit $b$, two valuations $val$ and $val^{\prime}$ and a program $\alpha$, the deterministic Boolean function $MOD$ taking as input a bit $b$, a valuation $val$ and a formula $\varphi$ and the deterministic Boolean function $ITE$ taking as input a bit $b$, two valuations $val$ and $val^{\prime}$, a program $\alpha$ and a sequence $d$ of $k$ bits.
Let $b$ be a bit, $val$ and $val^{\prime}$ be two valuations and $\alpha$ be a program.
The intuitive meaning of these functions will be explained later.
The deterministic Boolean function $REL$ is defined as follows:
%%%%%%%%%%%%%%%%%%%%%%%%%%%%%%%%%%%%%%%%%%%%%%%%%%%%%%%%%%%%%%%% REL BEGIN
\\
[0.20cm]
\textsf{function $REL(b,val,val^{\prime},\alpha)$ returns Boolean
\\
begin
\\
case $(b,\alpha)$ of
\\
\hspace*{0.5cm}$(0,+p)$:
%
%
%\\
%
%
%\hspace*{1.0cm}
$bool$ $:=$ ``$val^{\prime}$ $\not=$ $val\cup\{p\}$'';
\\
\hspace*{0.5cm}$(1,+p)$:
%
%
%\\
%
%
%\hspace*{1.0cm}
$bool$ $:=$ ``$val^{\prime}$ $=$ $val\cup\{p\}$'';
\\
\hspace*{0.5cm}$(0,-p)$:
%
%
%\\
%
%
%\hspace*{1.0cm}
$bool$ $:=$ ``$val^{\prime}$ $\not=$ $val\setminus\{p\}$'';
\\
\hspace*{0.5cm}$(1,-p)$:
%
%
%\\
%
%
%\hspace*{1.0cm}
$bool$ $:=$ ``$val^{\prime}$ $=$ $val\setminus\{p\}$'';
\\
\hspace*{0.5cm}$(0,\beta;\gamma)$:
\\
\hspace*{1.0cm}begin
\\
\hspace*{1.0cm}$bool$ $:=$ $true$;
\\
\hspace*{1.0cm}$val^{\prime\prime}$ $:=$ $0\ldots0$;
\\
\hspace*{1.0cm}repeat until $bool$ $=$ $false$ or  $val^{\prime\prime}$ $=$ $0\ldots0$
\\
\hspace*{1.5cm}begin
\\
\hspace*{1.5cm}$bool$ $:=$ $REL(0,val,val^{\prime\prime},\beta)$ or  $REL(0,val^{\prime\prime},val^{\prime},\gamma)$;
\\
\hspace*{1.5cm}$val^{\prime\prime}$ $:=$ $val^{\prime\prime}+1$
\\
\hspace*{1.5cm}end;
\\
\hspace*{1.0cm}end;
\\
\hspace*{0.5cm}$(1,\beta;\gamma)$:
\\
\hspace*{1.0cm}begin
\\
\hspace*{1.0cm}$bool$ $:=$ $false$;
\\
\hspace*{1.0cm}$val^{\prime\prime}$ $:=$ $0\ldots0$;
\\
\hspace*{1.0cm}repeat until $bool$ $=$ $true$ or  $val^{\prime\prime}$ $=$ $0\ldots0$
\\
\hspace*{1.5cm}begin
\\
\hspace*{1.5cm}$bool$ $:=$ $REL(1,val,val^{\prime\prime},\beta)$ and $REL(1,val^{\prime\prime},val^{\prime},\gamma)$;
\\
\hspace*{1.5cm}$val^{\prime\prime}$ $:=$ $val^{\prime\prime}+1$
\\
\hspace*{1.5cm}end;
\\
\hspace*{1.0cm}end;
\\
\hspace*{0.5cm}$(0,\beta\cup\gamma)$:
\\
\hspace*{1.0cm}$bool$ $:=$ $REL(0,val,val^{\prime},\beta)$ and $REL(0,val,val^{\prime},\gamma)$;
\\
\hspace*{0.5cm}$(1,\beta\cup\gamma)$:
\\
\hspace*{1.0cm}$bool$ $:=$ $REL(1,val,val^{\prime},\beta)$ or  $REL(1,val,val^{\prime},\gamma)$;
\\
\hspace*{0.5cm}$(0,\beta^{\star})$:
\\
\hspace*{1.0cm}begin
\\
\hspace*{1.0cm}$bool$ $:=$ $true$;
\\
\hspace*{1.0cm}$d$ $:=$ $0\ldots0$;
\\
\hspace*{1.0cm}repeat until $bool$ $=$ $false$ or  $d$ $=$ $0\ldots0$
\\
\hspace*{1.5cm}begin
\\
\hspace*{1.5cm}$bool$ $:=$ $ITE(0,val,val^{\prime},\beta,d)$;
\\
\hspace*{1.5cm}$d$ $:=$ $d+1$
\\
\hspace*{1.5cm}end;
\\
\hspace*{1.0cm}end;
\\
\hspace*{0.5cm}$(1,\beta^{\star})$:
\\
\hspace*{1.0cm}begin
\\
\hspace*{1.0cm}$bool$ $:=$ $false$;
\\
\hspace*{1.0cm}$d$ $:=$ $0\ldots0$;
\\
\hspace*{1.0cm}repeat until $bool$ $=$ $true$ or  $d$ $=$ $0\ldots0$
\\
\hspace*{1.5cm}begin
\\
\hspace*{1.5cm}$bool$ $:=$ $ITE(1,val,val^{\prime},\beta,d)$;
\\
\hspace*{1.5cm}$d$ $:=$ $d+1$
\\
\hspace*{1.5cm}end;
\\
\hspace*{1.0cm}end;
\\
\hspace*{0.5cm}$(0,\phi?)$:
\\
\hspace*{1.0cm}$bool$ $:=$ ``$val$ $\not=$ $val^{\prime}$'' or  $MOD(0,val^{\prime},\phi)$;
\\
\hspace*{0.5cm}$(1,\phi?)$:
\\
\hspace*{1.0cm}$bool$ $:=$ ``$val$ $=$ $val^{\prime}$'' and $MOD(1,val^{\prime},\phi)$
\\
end case;
\\
return $bool$
\\
end;}
%
%
%\\
%
%
%[0.20cm]
%
%

%%%%%%%%%%%%%%%%%%%%%%%%%%%%%%%%%%%%%%%%%%%%%%%%%%%%%%%%%%%%%%%% REL END

\medskip
Let $b$ be a bit, $val$ be a formula and $\varphi$ be a formula.
The deterministic Boolean function $MOD$ is defined as follows:
%%%%%%%%%%%%%%%%%%%%%%%%%%%%%%%%%%%%%%%%%%%%%%%%%%%%%%%%%%%%%%%% MOD BEGIN
%
%
\\
[0.20cm]
\textsf{function $MOD(b,val,\varphi)$ returns Boolean
\\
begin
\\
case $(b,\varphi)$ of
\\
\hspace*{0.5cm}$(0,p)$:
%
%
%\\
%
%
%\hspace*{1.0cm}
$bool$ $:=$ ``$p$ $\not\in$ $val$'';
\\
\hspace*{0.5cm}$(1,p)$:
%
%
%\\
%
%
%\hspace*{1.0cm}
$bool$ $:=$ ``$p$ $\in$ $val$'';
\\
\hspace*{0.5cm}$(0,\bot)$:
%
%
%\\
%
%
%\hspace*{1.0cm}
$bool$ $:=$ $true$;
\\
\hspace*{0.5cm}$(1,\bot)$:
%
%
%\\
%
%
%\hspace*{1.0cm}
$bool$ $:=$ $false$;
\\
\hspace*{0.5cm}$(0,\lbrack\alpha\rbrack\phi)$:
\\
\hspace*{1.0cm}begin
\\
\hspace*{1.0cm}$bool$ $:=$ $false$;
\\
\hspace*{1.0cm}$val^{\prime}$ $:=$ $0\ldots0$;
\\
\hspace*{1.0cm}repeat until $bool$ $=$ $true$ or  $val^{\prime}$ $=$ $0\ldots0$
\\
\hspace*{1.5cm}begin
\\
\hspace*{1.5cm}$bool$ $:=$ $REL(1,val,val^{\prime},\alpha)$ and $MOD(0,val^{\prime},\phi)$;
\\
\hspace*{1.5cm}$val^{\prime}$ $:=$ $val^{\prime}+1$
\\
\hspace*{1.5cm}end;
\\
\hspace*{1.0cm}end;
\\
\hspace*{0.5cm}$(1,\lbrack\alpha\rbrack\phi)$:
\\
\hspace*{1.0cm}begin
\\
\hspace*{1.0cm}$bool$ $:=$ $true$;
\\
\hspace*{1.0cm}$val^{\prime}$ $:=$ $0\ldots0$;
\\
\hspace*{1.0cm}repeat until $bool$ $=$ $false$ or  $val^{\prime}$ $=$ $0\ldots0$
\\
\hspace*{1.5cm}begin
\\
\hspace*{1.5cm}$bool$ $:=$ $REL(0,val,val^{\prime},\alpha)$ or  $MOD(1,val^{\prime},\phi)$;
\\
\hspace*{1.5cm}$val^{\prime}$ $:=$ $val^{\prime}+1$
\\
\hspace*{1.5cm}end;
\\
\hspace*{1.0cm}end;
\\
end case;
\\
return $bool$
\\
end;}
%
%
%\\
%
%
%[0.20cm]
%
%

%%%%%%%%%%%%%%%%%%%%%%%%%%%%%%%%%%%%%%%%%%%%%%%%%%%%%%%%%%%%%%%% MOD END

\medskip
Let $b$ be a bit, $val$ and $val^{\prime}$ be two valuations, $\alpha$ be a program and $d$ a sequence of $k$ bits.
The deterministic Boolean function $ITE$ is defined as follows:
%%%%%%%%%%%%%%%%%%%%%%%%%%%%%%%%%%%%%%%%%%%%%%%%%%%%%%%%%%%%%%%% ITE BEGIN
%
%
\\
[0.20cm]
\textsf{function $ITE(b,val,val^{\prime},\alpha,d)$ returns Boolean
\\
begin
\\
case $(b,d)$ of
\\
\hspace*{0.5cm}$(0,0\ldots0)$:
%
%
%\\
%
%
%\hspace*{1.0cm}
$bool$ $:=$ ``$val$ $\not=$ $val^{\prime}$'';
\\
\hspace*{0.5cm}$(1,0\ldots0)$:
%
%
%\\
%
%
%\hspace*{1.0cm}
$bool$ $:=$ ``$val$ $=$ $val^{\prime}$'';
\\
\hspace*{0.5cm}$(0,$odd integer$)$:
\\
\hspace*{1.0cm}begin
\\
\hspace*{1.0cm}$bool$ $:=$ $true$;
\\
\hspace*{1.0cm}$val^{\prime\prime}$ $:=$ $0\ldots0$;
\\
\hspace*{1.0cm}repeat until $bool$ $=$ $false$ or $val^{\prime\prime}$ $=$ $0\ldots0$
\\
\hspace*{1.5cm}begin
\\
\hspace*{1.5cm}$bool$ $:=$ $REL(0,val,val^{\prime\prime},\alpha)$ or  $ITE(0,val^{\prime\prime},val^{\prime},\alpha,d-1)$;
\\
\hspace*{1.5cm}$val^{\prime\prime}$ $:=$ $val^{\prime\prime}+1$
\\
\hspace*{1.5cm}end;
\\
\hspace*{1.0cm}end;
\\
\hspace*{0.5cm}$(1,$odd integer$)$:
\\
\hspace*{1.0cm}begin
\\
\hspace*{1.0cm}$bool$ $:=$ $false$;
\\
\hspace*{1.0cm}$val^{\prime\prime}$ $:=$ $0\ldots0$;
\\
\hspace*{1.0cm}repeat until $bool$ $=$ $true$ or  $val^{\prime\prime}$ $=$ $0\ldots0$
\\
\hspace*{1.5cm}begin
\\
\hspace*{1.5cm}$bool$ $:=$ $REL(1,val,val^{\prime\prime},\alpha)$ and $ITE(1,val^{\prime\prime},val^{\prime},\alpha,d-1)$;
\\
\hspace*{1.5cm}$val^{\prime\prime}$ $:=$ $val^{\prime\prime}+1$
\\
\hspace*{1.5cm}end;
\\
\hspace*{1.0cm}end;
\\
\hspace*{0.5cm}$(0,$even integer$)$:
\\
\hspace*{1.0cm}begin
\\
\hspace*{1.0cm}$bool$ $:=$ $true$;
\\
\hspace*{1.0cm}$val^{\prime\prime}$ $:=$ $0\ldots0$;
\\
\hspace*{1.0cm}repeat until $bool$ $=$ $false$ or  $val^{\prime\prime}$ $=$ $0\ldots0$
\\
\hspace*{1.5cm}begin
\\
\hspace*{1.5cm}$bool$ $:=$ $ITE(0,val,val^{\prime\prime},\alpha,d/2)$ or  $ITE(0,val^{\prime\prime},val^{\prime},\alpha,d/2)$;
\\
\hspace*{1.5cm}$val^{\prime\prime}$ $:=$ $val^{\prime\prime}+1$
\\
\hspace*{1.5cm}end;
\\
\hspace*{1.0cm}end;
\\
\hspace*{0.5cm}$(1,$even integer$)$:
\\
\hspace*{1.0cm}begin
\\
\hspace*{1.0cm}$bool$ $:=$ $false$;
\\
\hspace*{1.0cm}$val^{\prime\prime}$ $:=$ $0\ldots0$;
\\
\hspace*{1.0cm}repeat until $bool$ $=$ $true$ or  $val^{\prime\prime}$ $=$ $0\ldots0$
\\
\hspace*{1.5cm}begin
\\
\hspace*{1.5cm}$bool$ $:=$ $ITE(1,val,val^{\prime\prime},\alpha,d/2)$ and $ITE(1,val^{\prime\prime},val^{\prime},\alpha,d/2)$;
\\
\hspace*{1.5cm}$val^{\prime\prime}$ $:=$ $val^{\prime\prime}+1$
\\
\hspace*{1.5cm}end;
\\
\hspace*{1.0cm}end;
\\
end case;
\\
return $bool$
\\
end;}
%
%
%\\
%
%
%[0.20cm]
%
%
%
%

%%%%%%%%%%%%%%%%%%%%%%%%%%%%%%%%%%%%%%%%%%%%%%%%%%%%%%%%%%%%%%%% ITE END

\medskip
The deterministic Boolean function $REL$ takes as input a bit $b$, two valuations $val$ and $val^{\prime}$ and a program $\alpha$.
Its termination guarantees the following:

\begin{itemize}
\item if $REL(b,val,val^{\prime},\alpha)$ returns ``true'', then either $b=0$ and $(val,val^{\prime})\not\in\parallel\alpha\parallel$, or $b=1$ and $(val,val^{\prime})\in\parallel\alpha\parallel$,

\item if $REL(b,val,val^{\prime},\alpha)$ returns ``false'', then either $b=0$ and $(val,val^{\prime})\in\parallel\alpha\parallel$, or $b=1$ and $(val,val^{\prime})\not\in\parallel\alpha\parallel$.
\end{itemize}
The deterministic Boolean function $MOD$ takes as input a bit $b$, a valuation $val$ and a formula $\varphi$.
Its termination should guarantee the following:

\begin{itemize}
\item if $MOD(b,val,\alpha)$ returns ``true'', then either $b=0$ and $val\not\in\parallel\varphi\parallel$, or $b=1$ and $val\in\parallel\varphi\parallel$,
\item if $MOD(b,val,\alpha)$ returns ``false'', then either $b=0$ and $val\in\parallel\varphi\parallel$, or $b=1$ and $val\not\in\parallel\varphi\parallel$.
\end{itemize}
The deterministic Boolean function $ITE$ takes as input a bit $b$, two valuations $val$ and $val^{\prime}$, a program $\alpha$ and a sequence $d$ of $k$ bits. We identify the sequence $d$ and the integer represented by $d$.
Its termination should guarantee the following:
\begin{itemize}
\item if $ITE(b,val,val^{\prime},\alpha,d)$ returns ``true'', then either $b=0$ and $(val,val^{\prime})\not\in{\parallel\alpha^{d}\parallel}$, or $b=1$ and $(val,val^{\prime})\in\parallel\alpha^{d}\parallel$,
\item if $ITE(b,val,val^{\prime},\alpha,d)$ returns ``false'', then either $b=0$ and $(val,val^{\prime})\in{\parallel\alpha^{d}\parallel}$, or $b=1$ and $(val,val^{\prime})\not\in\parallel\alpha^{d}\parallel$.
\end{itemize}

\subsection{Soundness and completeness}
Let $\Gamma = \mathsf{Pgm}(PV)\times \mathsf{Fml}(PV)\times \overline{K}$ where $\overline{K}$ is the set of all sequences of $k$ bits.
We define the binary relation $\ll$ on $\Gamma$ in the following way:
$(\alpha,\phi,d)\ll(\beta,\psi,e)$ iff one of following condition holds:
\begin{itemize}
\item $len(\alpha)+len(\phi)<len(\beta)+len(\psi)$,
\item $len(\alpha)+len(\phi)=len(\beta)+len(\psi)$ and $d<e$.
\end{itemize}

\begin{lemma}\label{lell}
$\ll$ is a well-founded strict partial order on $\Gamma$.
\end{lemma}
\begin{proof}
By the well-foundedness of the standard linear order between non-negative integers.
\end{proof}

Let $\Sigma$ be the set of all $(\alpha,\phi,d)\in \Gamma$ such that the following condition holds:
\begin{enumerate}

\item for all bits $b$ and for all valuations $val$ and $val^{\prime}$,
\begin{itemize}
\item if $REL(b,val,val^{\prime},\alpha)$ returns ``true'', then either $b=0$ and $(val,val^{\prime})\not\in{\parallel\alpha\parallel}$, or $b=1$ and $(val,val^{\prime})\in\parallel\alpha\parallel$,
\item if $REL(b,val,val^{\prime},\alpha)$ returns ``false'', then either $b=0$ and $(val,val^{\prime})\in{\parallel\alpha\parallel}$, or $b=1$ and $(val,val^{\prime})\not\in\parallel\alpha\parallel$,
\end{itemize}

\item for all bits $b$ and for all valuations $val$,
\begin{itemize}
\item if $MOD(b,val,\phi)$ returns ``true'', then either $b=0$ and $val\not\in\parallel\phi\parallel$, or $b=1$ and $val\in\parallel\phi\parallel$,
\item if $MOD(b,val,\phi)$ returns ``false'', then either $b=0$ and $val\in\parallel\phi\parallel$, or $b=1$ and $val\not\in\parallel\phi\parallel$,
\end{itemize}

\item for all bits $b$ and for all valuations $val$ and $val^{\prime}$,
\begin{itemize}
\item if $ITE(b,val,val^{\prime},\alpha,d)$ returns ``true'', then either $b=0$ and $(val,val^{\prime})\not\in\parallel\alpha^{d}\parallel$, or $b=1$ and $(val,val^{\prime})\in\parallel\alpha^{d}\parallel$,
\item if $ITE(b,val,val^{\prime},\alpha,d)$ returns ``false'', then either $b=0$ and $(val,val^{\prime})\in\parallel\alpha^{d}\parallel$, or $b=1$ and $(val,val^{\prime})\not\in\parallel\alpha^{d}\parallel$.
\end{itemize}

\end{enumerate}

The aim is to prove by $\ll$-induction that all $(\alpha,\phi,d)$ are in $\Sigma$. As lemma \ref{lell} states that $\ll$ is a well-founded strict partial order, it is sufficient to prove the following lemma.

\begin{lemma}\label{induction}
Let $(\alpha,\phi,d)\in \Gamma$. If

\begin{center}
for all $(\beta,\psi,e)\in \Gamma$, if $(\beta,\psi,e)\ll(\alpha,\phi,d)$, then $(\beta,\psi,e)\in \Sigma$ \hspace{5mm} \emph{(H)}
\end{center}
then
$(\alpha,\phi,d)\in \Sigma$.
\end{lemma}
\begin{proof}
Suppose (H).\\
{\bf $(1.)$~The function $REL$.}
Let $b$ be a bit and $val$ and $val^{\prime}$ be valuations.
\\
Suppose $REL(b,val,val^{\prime},\alpha)$ returns ``true''.
We have to consider different cases.
\\
{\bf Cases $(b,\alpha)=(0,+p)$, or $(b,\alpha)=(1,+p)$, or $(b,\alpha)=(0,-p)$ and $(b,\alpha)=(1,-p)$.}
Left to the reader.
\\
{\bf Case $(b,\alpha)=(0,\beta;\gamma)$.}
Hence, $b=0$ and $\alpha=\beta;\gamma$.
Since $REL(b,val,val^{\prime},\alpha)$ returns ``true'', then for all valuations $val^{\prime\prime}$, either $REL(0,val,val^{\prime\prime},\beta)$ returns ``true'', or $REL(0,val^{\prime\prime},val^{\prime},\gamma)$ returns ``true''.
Remark that $(\beta,\phi,d)\ll(\alpha,\phi,d)$ and $(\gamma,\phi,d)\ll(\alpha,\phi,d)$.
Since (H), then $(\beta,\phi,d)\in \Sigma$ and $(\gamma,\phi,d)\in \Sigma$.
Since for all valuations $val^{\prime\prime}$, either $REL(0,val,val^{\prime\prime},\beta)$ returns ``true'', or\linebreak $REL(0,val^{\prime\prime},val^{\prime},\gamma)$ returns ``true'', then for all valuations $val^{\prime\prime}$, either\linebreak $(val,val^{\prime\prime})\not\in\parallel\beta\parallel$, or $(val^{\prime\prime},val^{\prime})\not\in\parallel\gamma\parallel$.
Thus, $(val,val^{\prime})\not\in\parallel\alpha\parallel$.
\\
{\bf Case $(b,\alpha)=(1,\beta;\gamma)$.}
Hence, $b=1$ and $\alpha=\beta;\gamma$.
Since $REL(b,val,val^{\prime},\alpha)$ returns ``true'', then there exists a valuation $val^{\prime\prime}$ such that $REL(b,val,val^{\prime\prime},\beta)$ returns ``true'' and $REL(b,val^{\prime\prime},val^{\prime},\gamma)$ returns ``true''.
Remark that $(\beta,\phi,d)\ll(\alpha,\phi,d)$ and $(\gamma,\phi,d)\ll(\alpha,\phi,d)$.
Since (H), then $(\beta,\phi,d)\in \Sigma$ and $(\gamma,\phi,d)\in \Sigma$.
Since there exists a valuation $val^{\prime\prime}$ such that $REL(1,val,val^{\prime\prime},\beta)$ returns ``true'' and $REL(1,val^{\prime\prime},val^{\prime},\gamma)$ returns ``true'', then there exists a valuation $val^{\prime\prime}$ such that $(val,val^{\prime\prime})\in\parallel\beta\parallel$ and $(val^{\prime\prime},val^{\prime})\in\parallel\gamma\parallel$.
Thus, $(val,val^{\prime})\in\parallel\alpha\parallel$.
\\
{\bf Cases $(b,\alpha)=(0,\beta\cup\gamma)$ and $(b,\alpha)=(1,\beta\cup\gamma)$.}
These cases are similarly treated.
\\
{\bf Case $(b,\alpha)=(0,\beta^{\star})$.}
Hence, $b=0$ and $\alpha=\beta^{\star}$.
Since $REL(b,val,val^{\prime},\alpha)$ returns ``true'', then for all sequences $e$ of $k$ bits, $ITE(0,val,val^{\prime},\beta,e)$ returns ``true''.
Remark that $(\beta,\phi,e)\ll(\alpha,\phi,d)$.
Since (H), then $(\beta,\phi,e)\in \Sigma$.
Since for all sequences $e$ of $k$ bits, $ITE(0,val,val^{\prime},\beta,e)$ returns ``true'', then for all sequences $e$ of $k$ bits, $(val,val^{\prime})\not\in\parallel\beta^{e}\parallel$.
Thus, $(val,val^{\prime})\not\in\parallel\alpha\parallel$.
\\
{\bf Case $(b,\alpha)=(1,\beta^{\star})$.}
Hence, $b=1$ and $\alpha=\beta^{\star}$.
Since $REL(b,val,val^{\prime},\alpha)$ returns ``true'', then there exists a sequence $e$ of $k$ bits such that\linebreak $ITE(1,val,val^{\prime},\beta,e)$ returns ``true''.
Remark that $(\beta,\phi,e)\ll(\alpha,\phi,d)$.
Since (H), then $(\beta,\phi,e)\in \Sigma$.
Since there exists a sequence $e$ of $k$ bits such that $ITE(1,val,val^{\prime},\beta,e)$ returns ``true'', then there exists a sequence $e$ of $k$ bits such that $(val,val^{\prime})\in\parallel\beta^{e}\parallel$.
Thus, $(val,val^{\prime})\in\parallel\alpha\parallel$.
\\
{\bf Case $(b,\alpha)=(0,\psi?)$.}
Hence, $b=0$ and $\alpha=\psi?$.
Since $REL(b,val,val^{\prime},\alpha)$ returns ``true'', then either $val\not=val^{\prime}$, or $MOD(0,val^{\prime},\psi)$ returns ``true''.
In the former case, $(val,val^{\prime})\not\in\parallel\alpha\parallel$.
In the latter case, remark that $(+p,\psi,d)\ll(\alpha,\phi,d)$.
Since (H), then $(+p,\psi,d)\in \Sigma$.
Since $MOD(0,val^{\prime},\psi)$ returns ``true''', then $val^{\prime}\not\in\parallel\psi\parallel$.
Thus, $(val,val^{\prime})\not\in\parallel\alpha\parallel$.
\\
{\bf Case $(b,\alpha)=(1,\psi?)$.}
Hence, $b=1$ and $\alpha=\psi?$.
Since $REL(b,val,val^{\prime},\alpha)$ returns ``true'', then $val=val^{\prime}$ and $MOD(1,val^{\prime},\psi)$ returns ``true''.
Remark that $(+p,\psi,d)\ll(\alpha,\phi,d)$.
Since (H), then $(+p,\psi,d)\in \Sigma$.
Since $MOD(1,val^{\prime},\psi)$ returns ``true''', then $val^{\prime}\in\parallel\psi\parallel$.
Since $val=val^{\prime}$, then $(val,val^{\prime})\in\parallel\alpha\parallel$.
\\
Suppose $REL(b,val,val^{\prime},\alpha)$ returns ``false''.
We have to consider cases similar to the above ones.
\\
{\bf $(2.)$~The function $MOD$.}
Let $b$ be a bit and $val$ be a valuation.
\\
Suppose $MOD(b,val,\phi)$ returns ``true''.
We have to consider several cases.
\\
{\bf Cases $(b,\phi)=(0,p)$, or $(b,\phi)=(1,p)$, or $(b,\phi)=(0,\bot)$ and $(b,\phi)=(1,\bot)$.}
Left to the reader.
\\
{\bf Case $(b,\phi)=(0,\lbrack\beta\rbrack\psi)$.}
Hence, $b=0$ and $\phi=\lbrack\beta\rbrack\psi$.
Since $MOD(b,val,\phi)$ returns ``true'', then there exists a valuation $val^{\prime}$ such that $REL(1,val,val^{\prime},\beta)$ returns ``true'' and $MOD(0,val^{\prime},\psi)$ returns ``true''.
Remark that $(\beta,\psi,d)\ll(\alpha,\phi,d)$.
Since (H), then $(\beta,\psi,d)\in \Sigma$.
Since there exists a valuation $val^{\prime}$ such that $REL(1,val,val^{\prime},\beta)$ returns ``true'' and $MOD(0,val^{\prime},\psi)$ returns ``true'', then there exists a valuation $val^{\prime}$ such that $(val,val^{\prime})\in\parallel\beta\parallel$ and $val^{\prime}\not\in\parallel\psi\parallel$.
Thus, $val\not\in\parallel\phi\parallel$.
\\
{\bf Case $(b,\phi)=(1,\lbrack\beta\rbrack\psi)$.}
Hence, $b=1$ and $\phi=\lbrack\beta\rbrack\psi$.
Since $MOD(b,val,\phi)$ returns ``true'', then for all valuations $val^{\prime}$, either\linebreak $REL(0,val,val^{\prime},\beta)$ returns ``true'', or $MOD(1,val^{\prime},\psi)$ returns ``true''.
Remark that $(\beta,\psi,d)\ll(\alpha,\phi,d)$.
Since (H), then $(\beta,\psi,d)\in \Sigma$.
Since for all valuations $val^{\prime}$, either $REL(0,val,val^{\prime},\beta)$ returns ``true'', or $MOD(1,val^{\prime},\psi)$ returns ``true'', then for all valuations $val^{\prime}$, either $(val,val^{\prime})\not\in\parallel\beta\parallel$, or $val^{\prime}\in\parallel\psi\parallel$.
Thus, $val\in\parallel\phi\parallel$.
\\
Suppose $MOD(b,val,\phi)$ returns ``false''.
We have to consider cases similar to the above ones.
\\
{\bf $(3.)$~The function $ITE$.}
Let $b$ be a bit and $val$ and $val^{\prime}$ be valuations.
\\
Suppose $ITE(b,val,val^{\prime},\alpha,d)$ returns ``true''.
We have to consider several cases.
\\
{\bf Cases $(b,d)=(0,0\ldots0)$, or $(b,d)=(1,0\ldots0)$.}
Left to the reader.
\\
{\bf Case $(b,d)=(0,$odd integer$)$.}
Hence, $b=0$ and $d=e1$ for some sequence $e$ of $k-1$ bits.
Since $ITE(b,val,val^{\prime},\alpha,d)$ returns ``true'', then for all valuations $val^{\prime\prime}$, either $REL(0,val,val^{\prime\prime},\alpha)$ returns ``true'', or $ITE(0,val^{\prime\prime},val^{\prime},\alpha,e)$ returns ``true''.
Remark that $(\alpha,\phi,e)\ll(\alpha,\phi,d)$.
Since (H), then $(\alpha,\phi,e)\in \Sigma$.
Since for all valuations $val^{\prime\prime}$, either $REL(0,val,val^{\prime\prime},\alpha)$ returns ``true'', or $ITE(0,val^{\prime\prime},val^{\prime},\alpha,e)$ returns ``true'', then for all valuations $val^{\prime\prime}$, either $(val,val^{\prime\prime})\not\in\parallel\alpha\parallel$, or $(val^{\prime\prime},val^{\prime})\not\in\parallel\alpha^{e}\parallel$.
Thus, $(val,val^{\prime})\not\in\parallel\alpha^{d}\parallel$.
\\
{\bf Case $(b,d)=(1,$odd integer$)$.}
Hence, $b=1$ and $d=e1$ for some sequence $e$ of $k-1$ bits.
Since $ITE(b,val,val^{\prime},\alpha,d)$ returns ``true'', then there exists a valuation $val^{\prime\prime}$ such that $REL(1,val,val^{\prime\prime},\alpha)$ returns ``true'' and\linebreak $ITE(1,val^{\prime\prime},val^{\prime},\alpha,e0)$ returns ``true''.
Remark that $(\alpha,\phi,e0)\ll(\alpha,\phi,d)$.
Since (H), then $(\alpha,\phi,e0)\in \Sigma$.
Since there exists a valuation $val^{\prime\prime}$ such that $REL(1,val,val^{\prime\prime},\alpha)$ returns ``true'' and $ITE(1,val^{\prime\prime},val^{\prime},\alpha,e0)$ returns ``true'', then there exists a valuation $val^{\prime\prime}$ such that $(val,val^{\prime\prime})\in\parallel\alpha\parallel$ and $(val^{\prime\prime},val^{\prime})\in\parallel\alpha^{e0}\parallel$.
Thus, $(val,val^{\prime})\in\parallel\alpha^{d}\parallel$.
\\
{\bf Case $(b,d)=(0,$even integer$)$.}
Hence, $b=0$ and $d=e0$ for some sequence $e$ of $k-1$ bits.
Since $ITE(b,val,val^{\prime},\alpha,d)$ returns ``true'', then for all valuations $val^{\prime\prime}$, either $ITE(0,val,val^{\prime\prime},\alpha,0e)$ returns ``true'', or $ITE(0,val^{\prime\prime},val^{\prime},\alpha,0e)$ returns ``true''.
Remark that $(\alpha,\phi,0e)\ll(\alpha,\phi,d)$.
Since (H), then $(\alpha,\phi,0e)\in \Sigma$.
Since for all valuations $val^{\prime\prime}$, either $ITE(0,val,val^{\prime\prime},\alpha,0e)$ returns ``true'', or $ITE(0,val^{\prime\prime},val^{\prime},\alpha,0e)$ returns ``true'', then for all valuations $val^{\prime\prime}$, either $(val,val^{\prime\prime})\not\in\parallel\alpha^{0e}\parallel$, or $(val^{\prime\prime},val^{\prime})\not\in\parallel\alpha^{0e}\parallel$.
Thus, $(val,val^{\prime})\not\in\parallel\alpha^{d}\parallel$.
\\
{\bf Case $(b,d)=(1,$even integer$)$.}
Hence, $b=1$ and $d=e0$ for some sequence $e$ of $k-1$ bits.
Since $ITE(b,val,val^{\prime},\alpha,d)$ returns ``true'', then there exists a valuation $val^{\prime\prime}$ such that $ITE(1,val,val^{\prime\prime},\alpha,0e)$ returns ``true'' and $ITE(1,val^{\prime\prime},val^{\prime},\alpha,0e)$ returns ``true''.
Remark that $(\alpha,\phi,0e)\ll(\alpha,\phi,d)$.
Since (H), then $(\alpha,\phi,0e)\in \Sigma$.
Since there exists a valuation $val^{\prime\prime}$ such that $ITE(1,val,val^{\prime\prime},\alpha,0e)$ returns ``true'' and $ITE(1,val^{\prime\prime},val^{\prime},\alpha,0e)$ returns ``true'', then there exists a valuation $val^{\prime\prime}$ such that $(val,val^{\prime\prime})\in\parallel\alpha^{0e}\parallel$ and $(val^{\prime\prime},val^{\prime})\in\parallel\alpha^{0e}\parallel$.
Thus, $(val,val^{\prime})\in\parallel\alpha^{d}\parallel$.
\\
Suppose $ITE(b,val,val^{\prime},\alpha,d)$ returns ``false''.
We have to consider cases similar to the above ones.
\end{proof}

\begin{proposition}\label{pro_principal}
$\Sigma = \Gamma$.
\end{proposition}

\begin{proof}
By Lemmas~\ref{lell} and~\ref{induction}.
\end{proof}
Hence, the functions $REL$, $MOD$ and $ITE$ are sound and complete.

\subsection{Complexity}
For all programs $\alpha$, let $f_{REL}(\alpha)$ be the maximal number of recursive calls between $REL$, $MOD$ and $ITE$ within the context of a call of the form\linebreak $REL(b,val,val^{\prime},\alpha)$.
For all formulas $\varphi$, let $f_{MOD}(\varphi)$ be the maximal number of recursive calls between $REL$, $MOD$ and $ITE$ within the context of a call of the form $MOD(b,val,\varphi)$.
For all programs $\alpha$, let $f_{ITE}(\alpha)$ be the maximal number of recursive calls between $REL$, $MOD$ and $ITE$ within the context of a call of the form $ITE(b,val,val^{\prime},\alpha,d)$.

\begin{claim}\label{claim:claim1}
$f_{ITE}(\alpha)$ $\leq$ $f_{REL}(\alpha)+2\times k-1$.
\end{claim}
\begin{proof}
Obvious.
\end{proof}

\begin{claim}\label{claim:claim2}
$f_{REL}(\alpha)$ $\leq$ $2\times len(\alpha)\times k$ and $f_{MOD}(\varphi)$ $\leq$ $2\times len(\varphi)\times k$.
\end{claim}

\begin{proof}
Let $\Pi$ be the property that holds for a pair $(\alpha,\varphi)$ iff $f_{REL}(\alpha)$ $\leq$ $2\times len(\alpha)\times k$ and $f_{MOD}(\varphi)$ $\leq$ $2\times len(\varphi)\times k$.
Let $\lldot$ be the binary relation that holds between pairs $(\alpha,\varphi)$ and $(\alpha^{\prime},\varphi^{\prime})$ iff either $len(\alpha)$ $<$ $len(\alpha^{\prime})$ and $len(\varphi)$ $\leq$ $len(\varphi^{\prime})$, or $len(\alpha)$ $\leq$ $len(\alpha^{\prime})$ and $len(\varphi)$ $<$ $len(\varphi^{\prime})$.
Remark that $\lldot$ is a well-founded order.
Let us demonstrate by $\lldot$-induction that $\Pi$ holds for all pairs $(\alpha, \varphi)$.
Let $(\alpha, \varphi)$ be such that for all $(\alpha^{\prime},\varphi^{\prime})$, if $(\alpha^{\prime},\varphi^{\prime})$ $\lldot$ $(\alpha,\varphi)$, then $\Pi$ holds for $(\alpha^{\prime}, \varphi^{\prime})$.
We only consider the following $2$ cases.
\\
{\bf Case $\alpha=\beta^{\star}$.}
Obviously, $f_{REL}(\beta^{\star})=f_{ITE}(\beta)+1$.
By Claim~\ref{claim:claim1}, $f_{ITE}(\beta)$ $\leq$ $f_{REL}(\beta)+2\times k-1$.
By induction hypothesis, $f_{REL}(\beta)$ $\leq$ $2\times len(\beta)\times k$.
Hence, $f_{REL}(\beta^{\star})$ $\leq$ $2\times(len(\beta)+1)\times k$ $\leq$ $2\times len(\beta^{\star})\times k$.
\\
{\bf Case $\varphi=\lbrack\beta\rbrack\phi$.}
Obviously, $f_{MOD}(\lbrack\beta\rbrack\phi)$ $\leq$ $\max\{f_{REL}(\beta),f_{MOD}(\phi)\}+1$.
By induction hypothesis, $f_{REL}(\beta)$ $\leq$ $2\times len(\beta)\times k$ and $f_{MOD}(\phi)$ $\leq$ $2\times len(\phi)\times k$.
Hence, $f_{MOD}(\lbrack\beta\rbrack\phi)$ $\leq$ $2\times\max\{len(\beta),len(\phi)\}\times k+1$ $\leq$ $2\times len(\lbrack\beta\rbrack\phi)\times k$.
\end{proof}

\medskip

Hence the maximal number of recursive calls between the deterministic Boolean functions $MOD$, $REL$ and $ITE$ has order linear in $k+len(\varphi)+len(\alpha)$.
Thus they can be implemented on deterministic Turing machines running in polynomial space.

This concludes the proof that our model checking algorithm works in polynomial space.
%This completes the proof of Proposition \ref{pro_pspace}. 

\section{Conclusion}
We have clarified the complexity of the model checking and the satisfiability problem of Dynamic Logic of Propositional Assignments (\DLPA) and of Coalition Logic of Propositional Control and Delegation \DCLPC.
First, we have explained why the proof of EXPTIME-hardness of the \DLPA model checking problem presented in \cite[Thm~$4$]{DBLP:conf/lics/BalbianiHT13} is erroneous. 
Second, although \DCLPC model checking is indeed in PSPACE, its proof in \cite[Thm.~$4$]{DBLP:journals/jair/HoekWW10} is flawed, and we have given a correct proof that the model checking and the satisfiability problem of both \DLPA and \DCLPC are in PSPACE.
All upper bounds are tight because the problem QSAT can be translated into the \DLPA model checking problem, as shown in \cite{DBLP:conf/ijcai/HerzigLMT11}.

%% The language of \DLPA can be extended with further program constructs.
%% The extension by the converse operator $(.)^-$ was already discussed in 
%% \cite[Section VII.B]{DBLP:conf/lics/BalbianiHT13}. 
%% Their proof that $\mathit{MC}$ and $\mathit{SAT}$ have the same complexity as the original language can be simplified: 
%% there is a linear reduction to the original language. 
%% Basically, the converse operator permutes with all program operators and 
%% can be eliminated when it faces a propositional variable by means of the equivalences
%% $ \lbrack (+p)^- \rbrack\phi \leftrightarrow \lbrack {p ?} ; ( {\top ?} \cup -p) \rbrack\phi  $
%% and 
%% $ \lbrack (-p)^- \rbrack\phi \leftrightarrow \lbrack \lnot p ? ; ( {\top ?} \cup +p) \rbrack\phi  $.

%% The language of \DLPA can be extended with the program construct $(\cdot\cap\cdot)$ of intersection.
%% In this variant, the value of $\alpha\cap\beta$ is the intersection of the values of $\alpha$ and $\beta$, i.e. $\parallel\alpha\cap\beta\parallel \ = \ \parallel\alpha\parallel\cap\parallel\beta\parallel$.
%% We believe that our deterministic Boolean functions $REL$, $MOD$ and $ITE$ can be appropriately modified.

%% \section*{Acknowledgements}

%% Special acknowledgement is heartly granted to Andreas Herzig who made several helpful comments for improving the correctness and the readability of this article.

\bibliographystyle{plain}
\bibliography{biblio-dlpa}

\begin{thebibliography}{1}

\bibitem{DBLP:conf/lics/BalbianiHT13}
Philippe Balbiani, Andreas Herzig, and Nicolas Troquard.
\newblock Dynamic logic of propositional assignments: A well-behaved variant of
  pdl.
\newblock In {\em 28th Annual ACM/IEEE Symposium on Logic in Computer Science},
  pages 143--152. IEEE Computer Society, 2013.

\bibitem{Chandra:1981:ALT:322234.322243}
Ashok~K. Chandra, Dexter~C. Kozen, and Larry~J. Stockmeyer.
\newblock Alternation.
\newblock {\em J. ACM}, 28(1):114--133, 1981.

\bibitem{AAMAS2014}
Olivier Gasquet, Valentin Goranko, and Fran{\c{c}}ois Schwarzentruber.
\newblock Big brother logic: logical modeling and reasoning about agents
  equipped with surveillance cameras in the plane.
\newblock In {\em International conference on Autonomous Agents and Multi-Agent
  Systems, {AAMAS} '14, Paris, France, May 5-9, 2014}, pages 325--332, 2014.

\bibitem{DBLP:conf/ijcai/HerzigLMT11}
Andreas Herzig, Emiliano Lorini, Fr{\'e}d{\'e}ric Moisan, and Nicolas Troquard.
\newblock A dynamic logic of normative systems.
\newblock In {\em IJCAI 2011, Proceedings of the 22nd International Joint
  Conference on Artificial Intelligence}, pages 228--233. IJCAI/AAAI, 2011.

\bibitem{Savitch1970177}
Walter~J. Savitch.
\newblock Relationships between nondeterministic and deterministic tape
  complexities.
\newblock {\em Journal of Computer and System Sciences}, 4(2):177 -- 192, 1970.

\bibitem{sipser2006introduction}
Michael Sipser.
\newblock {\em Introduction to the Theory of Computation}, volume~2.
\newblock Thomson Course Technology Boston, 2006.

\bibitem{stockemeyer-chandra-79}
Larry~J. Stockmeyer and Ashok~K. Chandra.
\newblock Provably difficult combinatorial games.
\newblock {\em SIAM Journal on Computing}, 8(2):151--174, 1979.

\bibitem{TiomkinMakowsky85}
Michael~L. Tiomkin and Johann~A. Makowsky.
\newblock Propositional dynamic logic with local assignments.
\newblock {\em Theor. Comput. Sci.}, 36:71--87, 1985.

\bibitem{DBLP:journals/jair/HoekWW10}
Wiebe van~der Hoek, Dirk Walther, and Michael Wooldridge.
\newblock Reasoning about the transfer of control.
\newblock {\em J. Artif. Intell. Res. (JAIR)}, 37:437--477, 2010.

\end{thebibliography}

\newpage
\section*{Appendix}

\begin{verbatim}
--- In (Balbiani, Herzig, Troquard, 2013 LICS) a supposedly polynomial
--- reduction from the problem PEEK-G5 (Stockemeyer Chandra 1979) into
--- the model checking problem in the logic of DL-PA is proposed. If
--- the reduction were actually working, a similar reduction could be
--- done from PEEK-G5 into the model checking problem of CTL over
--- NuSMV models.


--- We consider here the Peek instance where Eloise controls ep1, and
--- Abelard controls ap1 and ap2. Abelard plays first (Tau = A). The
--- goal formula Phi for this instance is ep1 & ap1. The valuation to
--- start with is empty: ep1, ap1, and ap2 are set to false. Clearly,
--- if Abelard never assigns true to ap1, Phi can never become
--- true. So clearly, Eloise does not have a winning strategy. So,
--- were the reduction working, we should not find a counter-model
--- when evaluating the present file. But a counter-model is found. So
--- the reduction in (Balbiani, Herzig, Troquard, 2013 LICS) is wrong.




MODULE abelard(turn, Phi)

--- Abelard controls two variables ap1 and ap2, both initially set to
--- false. Abelard can non-deterministically choose which variable to
--- change before his turn, that is, when it is the turn of
--- eloise. This is done by setting vartochange-a to either 1 or
--- 2. Then Abelard can set ap1 (next(ap1)) to either true or false,
--- when vartochange = 1, it is his turn (turn = a), and Phi is not
--- true. Abelard can set ap2 (next(ap2)) to either true or false, when
--- vartochange = 2, it is his turn (turn = a), and Phi is not true.

 VAR
  vartochange-a : {1,2};
  ap1 : boolean;
  ap2 : boolean;
 ASSIGN
  init(vartochange-a) := {1,2};
  init(ap1) := FALSE;
  init(ap2) := FALSE;
  next(vartochange-a) :=  (!Phi & turn = e) ? {1,2}: vartochange-a;  
  next(ap1) := (!Phi & turn = a & vartochange-a = 1) ? {TRUE, FALSE} : ap1;
  next(ap2) := (!Phi & turn = a & vartochange-a = 2) ? {TRUE, FALSE} : ap2;


MODULE eloise(turn, Phi)

--- Eloise controls only one variable ep1. Its initial value is set to
--- false. Eloise can set the value of ep1 (next(ep1)) to either true
--- or false, whenever it is her turn (turn = e) and Phi is not
--- true. Since she controls only one variable, the control variable
--- vartochange-e is dummy, but is used for uniformity with the MODULE
--- abelard.

 VAR
  vartochange-e : {1};
  ep1 : boolean;
 ASSIGN
  init(ep1) := FALSE;
  next(ep1) := (!Phi & turn = e & vartochange-e = 1) ? {TRUE, FALSE} : ep1;
  

MODULE main
 VAR
  turn : {e,a};
  nowin : boolean;

--- We consider here a Peek instance where Eloise controls ep1, and
--- Abelard controls ap1 and ap2. The valuation to start with is
--- empty: ep1, ap1, and ap2 are all set to false. In other words, elo
--- is an instance of the module eloise, and abe is an instance of the
--- module abelard; both defined in this file.

  elo : eloise(turn, Phi);
  abe : abelard(turn, Phi);

 DEFINE

--- In the Peek instance we consider, Abelard plays first (Tau =
--- A). The objective formula Phi is ep1 & ap1.   

Phi := (elo.ep1) & (abe.ap2); 
Tau := a;

 ASSIGN

  init(turn) := Tau; 
  init(nowin) := TRUE;
  
  next(turn) := 
case
         (turn = e) : a;
	 (turn = a) : e;
esac;

  next(nowin) := Phi ? FALSE : nowin;

CTLSPEC   

--- This formula is an immediate translation of the DL-PA formula in
--- (Balbiani, Herzig, Troquard, 2013 LICS) into the language of CTL.

   AG (nowin  ->  (
!Phi                     & 
((turn = e) -> AX nowin) & 
((turn = a) -> EX nowin)   
                  )
      )
\end{verbatim}

\end{document}